\documentclass[a4paper,10pt]{article}
\usepackage[utf8]{inputenc}
\usepackage{amsmath, amssymb, amsfonts, amsthm, bbm}
\usepackage{tikz-cd}
\usepackage[utf8]{inputenc}

\usepackage{geometry}
\usepackage{fancyhdr}
\usepackage{algorithm}
\usepackage{algpseudocode}
\usepackage[all]{xy}
\usepackage{braket}

\usepackage{biblatex}
\usepackage{amsmath}
\usepackage{amsfonts}
\usepackage{amssymb}
\usepackage[all]{xy}
\usepackage{graphicx}
\usepackage{authblk}
\usepackage{stix}
\newtheorem{dfn}{Definition}
\newtheorem{lem}{Lemma}
\newtheorem{thm}{Theorem}
\newtheorem{cor}{Corollary}
\newcommand{\bP}{\mathbb P}
\newcommand{\PP}{\mathbb{P}}

\title{On Multiquantum Bits, Segre Embeddings and Coxeter Chambers}
\author[1]{Noémie C. Combe}
\affil[1]{Department of Mathematics, University of Warsaw, Ulica Banacha 2, 02-097 Warszawa, n.combe@uw.edu.pl}

\begin{document}

\maketitle

\begin{abstract}

This work explores the interplay between quantum information theory, algebraic geometry, and number theory, with a particular focus on multiqubit systems, their entanglement structure, and their classification via geometric embeddings. The Segre embedding, a fundamental construction in algebraic geometry, provides an algebraic framework to distinguish separable and entangled states, encoding quantum correlations in projective geometry. We develop a systematic study of qubit moduli spaces, illustrating the geometric structure of entanglement through hypercube constructions and Coxeter chamber decompositions.

We establish a bijection between the Segre embeddings of tensor products of projective spaces and binary words of length 
$n-1$, structured as an $(n-1)$-dimensional hypercube, where adjacency corresponds to a single Segre operation. This reveals a combinatorial structure underlying the hierarchy of embeddings, with direct implications for quantum error correction schemes. The symmetry of the Segre variety under the Coxeter group of type 
$A$ allows us to analyze quantum states and errors through the lens of reflection groups, viewing separable states as lying in distinct Coxeter chambers on a Segre variety. The transitive action of the permutation group on these chambers provides a natural method for tracking errors in quantum states and potentially reversing them. Beyond foundational aspects, we highlight relations between Segre varieties and Dixon elliptic curves, drawing connections between entanglement and number theory. 

\textbf{Mathematics Subject Classification (MSC 2020):}
14M99 (None of the above, but in Algebraic Geometry),
20F55 (Reflection and Coxeter groups),
81P40 (Quantum coherence, entanglement, quantum correlations),
11G05 (Elliptic curves over global fields),
05C12 (Hypercubes in graph theory).
\end{abstract}
\thanks{
{\bf Acknowledgements} This research is part of the project No. 2022/47/P/ST1/01177 co-founded by the National Science Centre and the European Union's Horizon 2020 research and innovation program, under the Marie Sklodowska Curie grant agreement No. 945339 \includegraphics[width=1cm, height=0.5cm]{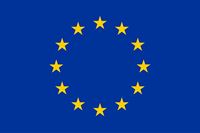}. I warmly thank Frederic Barbaresco for pointing out this passionating topic of qubits to me; 
Bartosz Naskrecki for stimulating and enthusiastic discussions about quantum computers; Bernd Sturmfels for drawing my attention to algebraic statistics; Piotr Sulkowski for organising the string theory club, where I learned many beautiful subjects;   Benjamin Collas for giving me the possibility of being part of the Arithmetic and Homotopic Galois Theory group RIMS-ENS; Jaroslaw Wisniewski for stimulating discussions on Coxeter groups. Finally, I thank Don Zagier for his great enthusiasm and support. 

\section{Introduction}
Mathematical databases form a growing area of research at the intersection of mathematics and computer science. These databases organize, store, and provide access to mathematical knowledge in a structured and searchable manner. 

\,

This paper seeks to reflect on the conceptual framework underlying these systems. We choose to consider quantum bits \cite{G}, which can encode an enormous amount of information, and investigate underlying structures. 

\,

An essential challenge, related to quantum bits emerges from the theory of entanglement. The problem lies in establishing robust criteria for separability. Although the notion of separability is precisely defined, the practical determination of whether a given quantum state is separable or entangled often eludes a straightforward resolution. \cite{BZ}

\, 

This difficulty naturally refines into a more nuanced problem: the \emph{ quantification and degree of entanglement} for a specific entangled state. Here, we encounter a vast and largely uncharted landscape, as there exists no canonical measure of entanglement. 

\, 

The choice of a metric often reflects the context of the problem, shaped by the underlying physical setup or the intended application. Thus, the study of entanglement, in its very essence, eludes a single, definitive mathematical formulation. Rather than being a fully developed mathematical theory, it remains an intricate and evolving question, standing at the crossroads of algebra, geometry, and computation. Its very nature suggests that multiple formalisms may coexist, each capturing different aspects of this deep and enigmatic structure.

\,

Another inherent difficulty with qubits lies in their susceptibility to errors. The subtle interplay between coherence and decoherence, entanglement and noise, renders precise control over these errors a tricky goal. It is precisely this challenge—the struggle to correct quantum errors—that serves as the driving force behind much of the ongoing research in this domain. We propose an attempt to simplify the error correction, based on \cite{M} and \cite{C18,C19}.

\,

In the elegant framework of complex algebraic geometry, the notion of separability for quantum states acquires a particularly natural and concise formulation. Pure multiparticle states correspond to points in the complex projective space \( \mathbb{P}^N \), which parametrizes lines through the origin in the complex vector space \( \mathbb{C}^{N+1} \). Entanglement, in this geometric setting, is illuminated through the Segre embedding, a fundamental construction arising from the categorical product of projective spaces. Explicitly, this is a morphism of complex algebraic varieties:  
\[
\mathbb{P}^{n_1} \times \cdots \times \mathbb{P}^{n_k} \longrightarrow \mathbb{P}^N,
\]  

where \( N = (n_1+1)(n_2+1)\cdots(n_k+1) - 1 \), that maps the Cartesian product of projective spaces into a single projective space.  

\,

The image of this embedding, known as the \emph{Segre variety}, serves as a geometric locus that characterizes separable states. Its structure is governed by a family of homogeneous quadratic polynomial equations in \( 2^n \) variables, where \( n \) denotes the number of particles. In this language, entanglement emerges as the deviation of a state from this algebraic variety---a property encoded in the intricate geometry of the embedding and the defining equations of the Segre variety.

\section{Quantum Bits and entanglements}

\subsubsection{Classical bits}
In the classical theory of computation and digital communication, the fundamental unit of information is the binary digit, or bit. 
Digital information is typically represented in two states: $0$ and $1$. Mathematically, a bit is an element of the set $\{0,1\}$. 

\, 
\subsubsection{Quantum bits}
 Their quantum analogs, the quantum bits (qubits), offer a powerful tool to explore: even a small number of qubits can encode an enormous amount of information because particles can be in superpositions of states. However, quantum bits are more susceptible to errors than their classical analogs.

\,

The fundamental unit of quantum information, the quantum bit or qubit, is represented by a complex vector:

\begin{equation}
\begin{bmatrix} \alpha \\ \beta \end{bmatrix}\in \mathbb{C}^2,
\end{equation}

where $\alpha$ and $\beta$ are complex numbers representing the weight of the $0$ and $1$ states of the superposition. Usually, we denote by $\mathcal{H}_2$ the corresponding 2-dimensional Hilbert space and when it is clear from the context we simply write $\mathcal{H}$. These coefficients are complex probability amplitudes and satisfy the normalization condition:

\begin{equation}
|\alpha|^2 + |\beta|^2 = 1.
\end{equation}

Notice that this is precisely the equation of a 3-sphere in $\mathbb{R}^4$.

\subsubsection{Dirac notation}

The state corresponding to the 0 (resp. 1) of the classical bit are written in Dirac's bra-ket notation as follows:

\begin{equation}
|0\rangle = \begin{bmatrix} 1 \\ 0 \end{bmatrix}, \quad |1\rangle = \begin{bmatrix} 0 \\ 1 \end{bmatrix}.
\end{equation}

Using Dirac’s bra-ket notation, a quantum bit is therefore expressed as:

\begin{equation}
|\psi\rangle =\alpha |0\rangle + \beta |1\rangle .
\end{equation}
\subsubsection{Qubit moduli space}
Following the previous subsection, one quibit is parametrized by a 3-sphere.  In the context of qubits, this sphere is often referred as the Bloch sphere. This is the moduli space parametrizing the space of states of the qubit and given by the complex equation $|\alpha|^2 + |\beta|^2 = 1$, where $\alpha, \beta\in \mathbb{C}$. Note that this sphere can be also interpreted as quaternion vectors of unit norm (quaterion versors) i.e. those elements  $q\in \mathbb{H}$ such that $|q|^2=1,$ where $|q|^2=qq^*$, where $q^*$ is the (quaternionic) conjugate.

\subsubsection{Operations on the the moduli space}
We can describe our state via a two-component complex vector. Any possible modification of that vector can be represented by a Hermitian $2\times 2$ matrix. 

Operations on the moduli space of a qubit are represented by Pauli operators. 

The Pauli matrices form a complete basis for any Hermitian $2 \times 2$ matrix:

\[R= \begin{pmatrix}
1&0\\
0&-1
\end{pmatrix} \quad I= \begin{pmatrix}
1&0\\
0&1
\end{pmatrix}\quad 
S=\begin{pmatrix}
0&1\\
1&0
\end{pmatrix} \quad 
V=\begin{pmatrix}
0&-\imath\\
\imath&0
\end{pmatrix} \]

\,

\subsubsection{Probability Simplex}

Regarding this from the probabilistic perspective, we have the probabilities $p_1=|\alpha|^2$ and $p_2=|\beta|^2$ where naturally $p_i\geq 0$ for $i\in\{1,2\}$  and $p_1+p_2=1$. 
We have done a change of variables mapping $$\begin{aligned}\alpha\mapsto | \alpha|^2= \alpha \alpha^*=p_1,\\ \beta\mapsto | \beta|^2= \beta \beta^*=p_2.\end{aligned}$$
The image of this map generates a probabilistic simplex of dimension 1, that we denote  $\Delta_1$.
In a more general framework, the standard simplex $\Delta_m$ in $\mathbb{R}^{m+1}$ is given by $$\Delta_m={\rm conv}\{e_0,e_1,\cdots,e_m\}\subset \mathbb{R}^{m+1}.$$

\subsection{Multiquantum bits and entanglement problem}
Whenever more than one qubit is involved in our system, the number of computational basis states increases: the number of states grows exponentially. Given $N$ qubits, there are $2^N$ basis states.
\subsubsection{Example}
\begin{itemize}
    \item For $N=2$ qubits, this generates four basis states:  
    \[
    |00\rangle, \quad |10\rangle, \quad |01\rangle, \quad |11\rangle.
    \]
    
    \item For $N=3$ qubits, we have eight basis states:  
    \[
    \begin{aligned}
    &|000\rangle, \quad |100\rangle, \quad |010\rangle, \quad |001\rangle, \\
    &|110\rangle, \quad |101\rangle, \quad |011\rangle, \quad |111\rangle.
    \end{aligned}
    \]
    
    \item For $N=4$ qubits, we have sixteen basis states:  
    \[
    \begin{aligned}
    &|0000\rangle, \quad |0100\rangle, \quad |0010\rangle, \quad |0001\rangle, \\
    &|0110\rangle, \quad |0101\rangle, \quad |0011\rangle, \quad |0111\rangle, \\
    &|1000\rangle, \quad |1100\rangle, \quad |1010\rangle, \quad |1001\rangle, \\
    &|1110\rangle, \quad |1101\rangle, \quad |1011\rangle, \quad |1111\rangle.
    \end{aligned}
    \]
\end{itemize}

A quantum state must specify the complex coefficients of all of these basis vectors; this information can no longer be represented as a simple geometrical picture like the Bloch sphere for more than 1 qubit.

\subsubsection{Two qubit pure states: entangled or separable states?}
 Given two qubits, their state lives in the tensor product space: $\mathcal{H}_A \otimes \mathcal{H}_B\cong \mathbb{C}^4$, which is a Hilbert space of four complex dimensions. 

\,

A general two-qubit state is given by:

\begin{equation}
| \psi \rangle = \alpha |00\rangle + \beta |01\rangle + \gamma |10\rangle + \delta |11\rangle,
\end{equation}
where the coefficients $\alpha, \beta,\gamma , \delta\in \mathbb{C}$ are complex numbers which satisfy
 \[|\alpha|^2+|\beta|^2+|\gamma|^2+|\delta|^2=1.\]

\begin{dfn}~
\begin{itemize}
\item A state is entangled if it cannot be factorized as a simple tensor product.
\item A separable (i.e. non-entangled) state takes the form:
\end{itemize}
\begin{equation}
| \psi \rangle = (c_{00} |0\rangle + c_{01} |1\rangle) \otimes (c_{10} |0\rangle + c_{11} |1\rangle).
\end{equation}
\end{dfn}
We will discuss a criterion allowing to distinguish the case of separable states from the entangled states. 
\section{Quantum Entanglement and the Segre Embedding} 
\subsubsection{Quantum Entanglement Problem}
One of the fundamental features of quantum mechanics is entanglement. An essential challenge in the theory of entanglement lies in establishing a good criterion for separability and obtaining information about the entanglement. Although the notion of separability is precisely defined, the practical determination of whether a given quantum state is separable or entangled often eludes a straightforward resolution. Here, we encounter a vast and largely uncharted landscape, as there exists \emph{no canonical measure} of entanglement.

\subsubsection{Segre embedding}

 The \emph{Segre embedding} provides a natural algebraic construction that relates the tensor product of vector spaces (for instance Hilbert spaces) with the product of their corresponding projective spaces. Explicitly, consider two finite-dimensional complex vector spaces \( V \) and \( W \), and their associated projective spaces \( \mathbb{P}(V) \) and \( \mathbb{P}(W) \). The Segre embedding is the canonical map:
\[
\sigma : \mathbb{P}(V) \times \mathbb{P}(W) \hookrightarrow \mathbb{P}(V \otimes W),
\]
given by:
\[
\sigma([v], [w]) = [v \otimes w],
\]
where \( [v] \in \mathbb{P}(V) \) and \( [w] \in \mathbb{P}(W) \) are projective points corresponding to nonzero vectors \( v \in V \) and \( w \in W \).

In projective coordinates, if we take  \(  [v] =[x_0: \dots : x_m] \in \mathbb{P}^m \)  and \([w]=[y_0: \dots : y_n]  \in  \mathbb{P}^n \) then the embedding is expressed as:
\[
\sigma : \mathbb{P}^m \times \mathbb{P}^n \hookrightarrow \mathbb{P}^{mn+m+n},
\]
\[
\sigma([x_0: \dots : x_m], [y_0: \dots : y_n]) \mapsto [x_0y_0 : x_0y_1 : \dots : x_my_n].
\]

The image of \( \sigma \), known as the \emph{Segre variety}, is characterized by the vanishing of all \( 2 \times 2 \) minors of the rank-one matrix:
\[
\begin{pmatrix}
x_0y_0 & x_0y_1 & \dots & x_0y_n \\
x_1y_0 & x_1y_1 & \dots & x_1y_n \\
\vdots & \vdots & \ddots & \vdots \\
x_my_0 & x_my_1 & \dots & x_my_n
\end{pmatrix}.
\]

This means that each coordinate $[z_{00}:\cdots: z_{mn}]\in \mathbb{P}^{mn+m+n}$ in the image of $\sigma$ is given by a monomial of the form: $z_{ij}=x_iy_j$. 

\subsubsection{Example}
We take the case 
\[
\sigma : \mathbb{P}^2 \times \mathbb{P}^2 \hookrightarrow \mathbb{P}^{8},
\]
\[
\sigma : ([x_0:x_1:x_2] ,[y_0:y_1:y_2]) \mapsto [x_0y_0:x_0y_1:\cdots:x_2y_2],
\]
Let 

\[M=\begin{bmatrix}
z_{00} & z_{01} & z_{02} \\
z_{10} & z_{11} & z_{12} \\
z_{20} & z_{21} & z_{22}
\end{bmatrix}
\]
be a rank-one  \(3 \times 3\) matrix, where $z_{ij}=x_iy_j$.

\textbf{The \(2 \times 2\) minors (determinants of \(2 \times 2\) submatrices):}
\[
\begin{aligned}
&\det \begin{bmatrix} z_{00} & z_{01} \\ z_{10} & z_{11} \end{bmatrix} = z_{00}z_{11} - z_{01}z_{10}, \\
&\det \begin{bmatrix} z_{00} & z_{02} \\ z_{10} & z_{12} \end{bmatrix} = z_{00}z_{12} - z_{02}z_{10}, \\
&\det \begin{bmatrix} z_{01} & z_{02} \\ z_{11} & z_{12} \end{bmatrix} = z_{01}z_{12} - z_{02}z_{11}, \\
&\det \begin{bmatrix} z_{00} & z_{01} \\ z_{20} & z_{21} \end{bmatrix} = z_{00}z_{21} - z_{01}z_{20}, \\
&\det \begin{bmatrix} z_{00} & z_{02} \\ z_{20} & z_{22} \end{bmatrix} = z_{00}z_{22} - z_{02}z_{20}, \\
&\det \begin{bmatrix} z_{01} & z_{02} \\ z_{21} & z_{22} \end{bmatrix} = z_{01}z_{22} - z_{02}z_{21}, \\
&\det \begin{bmatrix} z_{10} & z_{11} \\ z_{20} & z_{21} \end{bmatrix} = z_{10}z_{21} - z_{11}z_{20}, \\
&\det \begin{bmatrix} z_{10} & z_{12} \\ z_{20} & z_{22} \end{bmatrix} = z_{10}z_{22} - z_{12}z_{20}, \\
&\det \begin{bmatrix} z_{11} & z_{12} \\ z_{21} & z_{22} \end{bmatrix} = z_{11}z_{22} - z_{12}z_{21}.
\end{aligned}
\]

These equations define the Segre variety. 
\subsubsection{The particular case of $\mathbb{P}^1\times  \mathbb{P}^1$}
For $m=n=1$, the Segre map is:

 \[
\sigma([x_0: x_1], [y_0: y_1]) \mapsto [x_0y_0 : x_0y_1 : x_1y_0 : x_1y_1].
\]

The new coordinates satisfy the determinantal equation:
\[z_{00}z_{11}-z_{10}z_{01}=0\]
which defines a Segre variety in $\mathbb{P}^3$, and this generates a toric ideal. However, note that this is an exceptional case and that not all Segre varieties are necessarily toric. 

\subsubsection{Connection Between the Segre Variety and Dixon Elliptic Curves}
The case of $m=n=1$, also holds important relations to Dixon elliptic curves.
The Segre variety $S_{1,1}=\{z_{00}z_{11}-z_{10}z_{01}=0\}$ is a quadric surface, and Dixon elliptic curves are obtained by intersecting this quadric surface $S$ with another quadric surface. Namely, the Dixon elliptic curve is the intersection of 
S with a second quadric surface say $Q(z_{00},z_{01},z_{10},z_{11})=0$

Thus, Dixon elliptic curves are smooth complete intersections of the Segre variety with another quadric surface in the projective space $\mathbb{P}^3.$

\section{Separable and Entangled States}
We go back to our example of a two-qubit pure state. This  can be written as
\begin{equation}
    |\Psi\rangle = \sum_{i,j=0}^{1} c_{ij} |i\rangle |j\rangle,
\end{equation}
where $c_{ij}$ are the complex coefficients. 

\,

The homogeneous coordinates in $\mathbb{ P}^{3}$ are given by: $[Z^{0}: Z^{1}: Z^{2}: Z^{3}]$ where  $(Z^{0}, Z^{1}, Z^{2}, Z^{3})$ is identified with the four-tuple $(c_{00}, c_{01}, c_{10}, c_{11})$.

\begin{lem}
    Let \( |\psi\rangle_{AB} \) be a  two-qubit pure state. Let $\mathcal{H}_A$ and  $\mathcal{H}_B$  be the respective Hilbert spaces of the subsystems  $A$  and  $B$. 
    
    The pure state \( |\psi\rangle_{AB} \) is separable if and only if the corresponding coordinates 
    \[[Z^{0}: Z^{1}: Z^{2}: Z^{3}] \in \mathbb{P}^3\quad \text{of}\quad  |\psi\rangle_{AB} \] lie on a \emph{Segre variety}, which is the image of the canonical embedding:
\[
\mathbb{P}(\mathcal{H}_A) \times \mathbb{P}(\mathcal{H}_B) \hookrightarrow \mathbb{P}^3.
\]
Explicitly, the Segre variety is described as the locus of points that satisfy the quadratic equation:
\[
\det
\begin{pmatrix}
c_{00} & c_{01} \\
c_{10} & c_{11}
\end{pmatrix}
= 0,
\]
where \( |\psi\rangle_{AB} = c_{00}|00\rangle + c_{01}|01\rangle + c_{10}|10\rangle + c_{11}|11\rangle \).
\end{lem}
\begin{proof} 
One direction is easy. Indeed, it is straightforward to check by contraposition that if $[Z^{0}: Z^{1}: Z^{2}: Z^{3}]$ does not lie on a Segre variety then it cannot correspond to a pure separable state. This is essentially due to the construction. 

The other direction, can be done by contradiction. Assume that the coefficients of the vector corresponding to a pure separated state $|\psi\rangle$ do not lie on the Segre variety. Therefore, $c_{00}c_{11}-c_{01}c_{10}\neq0$. However, by construction $c_{00}=\alpha\gamma$, $c_{11}=\beta\delta$, $c_{01}=\alpha\delta$,$c_{10}=\beta\gamma$. By hypothesis, we have $(\alpha|0\rangle+\beta|1\rangle)\otimes(\gamma|0\rangle+\delta|1\rangle)$ which is equal to $\alpha|\gamma00\rangle+\alpha\delta|01\rangle+\beta\gamma |10\rangle+\beta\gamma |11\rangle)$. But having $c_{00}c_{11}-c_{01}c_{10}\neq0$ implies that 
$(\alpha\gamma)\cdot(\beta\delta)\neq(\alpha\delta)\cdot(\beta\gamma)$ which is a contradiction. 
\end{proof}

Therefore, 
\begin{cor}A state corresponding to a 2-qubit is separable if and only if
\begin{equation} Z^0Z^3 - Z^1Z^2 = 0.
\label{segre_condition}
\end{equation}
This equation defines the Segre variety in $\mathbb{P}^{3}$.
\end{cor}
So, pure separable states are precisely those that lie on this variety; entangled states correspond to points outside this locus.

The geometric nature of entanglement thus becomes apparent: an entangled state is one that deviates from the Segre variety, and its degree of entanglement can be analyzed through its projection onto this variety.
Therefore, separable states form a four (real) dimensional manifold of a six (real) dimensional space of all states. 

\,
\subsubsection{Geometric description}
We describe the previous construction furthermore. Consider a hyperoctant in a 3-sphere and the following coordinates: 

$(c_{00},c_{01},c_{01},c_{11})=(a_{00},a_{01}e^{\imath v_1},a_{01}e^{\imath v_2},a_{11}e^{\imath v_3}),$ where 
\[ a_{00}^2+\cdots +a_{11}^2=1\]

\,

Using a \emph{ gnomonic projection} of the 3-sphere  centered at:  $$(a_{00},a_{01},a_{10},a_{11})=\frac{1}{2}(1,1,1,1),$$

the Segre embedding can be nicely described in the octant picture.

The complex equation $c_{00}c_{11}-c_{01}c_{10}=0$ splits into 2 real equations:

\[
a_{00}a_{11}-a_{01}a_{10}=0 \]
\[ v_1+v_2-v_3=0    \]

Hence, one may envision the space of separable states as a two-dimensional algebraic surface embedded within the octant. Over each separable point in this octant, a corresponding two-dimensional surface resides in the torus.

\,

This concept can be generalized to $n$-qubit systems, where separability structures correspond to higher Segre embeddings. This exhibits altogether the structure of a hypercube of Segre embeddings. 

\subsection{Characterization of $n$-qubit Systems via Segre Embeddings and Hypercubes}

\subsubsection{Case: $N=3$}
\begin{lem}
The quantum state \( |\psi\rangle \) of a 3-qubit is said to be separable with respect to the tensor decomposition \( A \otimes B \otimes C \) if and only if its associated point \([ \psi ] \in \mathbb{P}^7\) lies on the generalized \emph{Segre variety}. This variety arises naturally as the image of the Segre embedding:
\[
\mathbb{P}(\mathcal{H}_A) \times \mathbb{P}(\mathcal{H}_B) \times \mathbb{P}(\mathcal{H}_C) \hookrightarrow \mathbb{P}(\mathcal{H}_A \otimes \mathcal{H}_B \otimes \mathcal{H}_C),
\]
where \( \mathcal{H}_A \), \( \mathcal{H}_B \), and \( \mathcal{H}_C \) are the respective Hilbert spaces of the subsystems \( A \), \( B \), and \( C \). 
\end{lem}
\begin{proof}
The proof is a straightforward computation, that mimics essentially the case where $N=2$. 
\end{proof}

We illustrate this using the following diagram: 
\[
\begin{tikzcd}
\bP^1 \times \bP^1 \times \bP^1  \arrow[r, "\sigma_{A,B}\otimes I_C"] \arrow[d, "I_A \otimes \sigma_{B,C}"'] & \bP^3 \times \bP^1  \arrow[d, "\sigma_{A,B\otimes C}"] \\
\bP^1 \times \bP^3   \arrow[r, "\sigma_{A,B\otimes C}"] & \bP^7 
\end{tikzcd}
\]

In this perspective, separability is geometrically encoded as the property that the coefficients of the state vector \( |\psi\rangle \) lies on the Segre variety rather than in its complement, which is geometrically interpreted as the space of entangled states.

\subsubsection{General case}
\begin{thm}

Let \( \mathcal{H}_2 \) denote the 2-dimensional Hilbert space of a single qubit. 

The pure state \( |\psi\rangle \) of a $n$-qubit is a separable state if and only if the associated coordinate points in $\mathbb{P}^{2^n - 1}$ lie on the generalized \emph{Segre variety}. Furthermore, it is a product state if it is parametrized by a \(n-1\)-hypercube of Segre embeddings given by:

\[
\mathbb{P}^1 \times \mathbb{P}^1 \times \dots \times \mathbb{P}^1 \hookrightarrow \mathbb{P}^{2^n - 1}.
\]

\end{thm}
\begin{proof}
\subsubsection*{Preliminary Notation and Inductive Construction}

Before proceeding with the proof, we introduce some notation. The symbol \(\sigma_{A,B}\) represents the Segre embedding of the projective spaces associated with the Hilbert spaces \(\mathbb{P}(\mathcal{H}_A)\) and \(\mathbb{P}(\mathcal{H}_B)\), given by:
\[
\sigma_{A,B}: \mathbb{P}(\mathcal{H}_A) \times \mathbb{P}(\mathcal{H}_B) \hookrightarrow \mathbb{P}(\mathcal{H}_A \otimes \mathcal{H}_B).
\]
Additionally, the symbol \(I\) denotes the identity morphism.

\subsubsection*{Inductive Construction of the Embedding Cube - low dimensions}

The proof of this construction follows by induction. The embedding cube is constructed inductively as follows.

\( 1.\) The Segre embedding \[\mathbb{P}^1 \times \mathbb{P}^1  \hookrightarrow \mathbb{P}^3\] illustrates an edge connecting two vertices of the embedding diagram. This is therefore a 1-cube (a cube of dimension one). 

\, 
\( 2.\) Suppose that we have now $n=3$, that is:  \[\underbrace{\bP^1 \times \bP^1 \times \bP^1}_{3}  \hookrightarrow \mathbb{P}^7.\] 
Using the previous step \( 1\) to compute the Segre embedding diagram yields a 2-cube structure of embeddings:
\[
\xymatrix{
\PP^1_A\times \PP^1_B\times \PP^1_C\ar[d]_{{I}_A\otimes \sigma_{B,C}}
\ar[rr]^{\sigma_{A,B}\otimes {I}_C}&&\PP^3_{AB}\times \PP^1_C\ar[d]^{\sigma_{A,B\otimes C}}\\
\PP^1_A\times \PP^3_{BC}\ar[rr]^{\sigma_{A\otimes BC}}&&\PP^7_{ABC}
}.
\]

\, 
\( 3.\) Suppose that we have now $n=4$, that is we have \[\underbrace{\bP^1 \times \cdots \times \bP^1}_{4}  \hookrightarrow \mathbb{P}^{15}.\] 
Using the previous step \( 2\) to compute the Segre embedding diagram yields a 3-cube structure of embeddings:
\[ \xymatrix{ \PP^1\times\PP^1\times\PP^1\times\PP^1 \ar[dd]_{\sigma_{1,1}\times I\times I}\ar[rd]^{I\times \sigma_{1,1}\times {I}} \ar[rr]^{I\times I\times \sigma_{1,1}} && 
\PP^1\times\PP^1\times\PP^3 \ar'[d][dd]^{\sigma_{1,1}\times I} \ar[rd]^{I\times \sigma_{1,3}} \\
& \PP^1\times\PP^3\times\PP^1 \ar[dd]_(.3){\sigma_{1,3}\times I} \ar[rr]^(.4){ I\times \sigma_{3,1}} && \PP^1\times\PP^7 \ar[dd]^{\sigma_{1,7}} \\
\PP^3\times\PP^1\times\PP^1 \ar'[r][rr]^(.2){ I\times \sigma_{1,1}} \ar[rd]^{\sigma_{3,1}\times I}  && \PP^3\times\PP^3 \ar[rd]^{\sigma_{3,3}} \\
& \PP^7\times\PP^1 \ar[rr]^{\sigma_{7,1}} && \PP^{15} }.
\]

\subsubsection*{Segre Embedding as a 4-Cube Diagram for \( n=5 \)}

\( 4.\) Suppose that we have \( n=5 \), i.e., 
\[
\underbrace{\mathbb{P}^1 \times \cdots \times \mathbb{P}^1}_{5} \hookrightarrow \mathbb{P}^{31}.
\]

Applying our inductive construction, the Segre embedding diagram follows the structure of a 4-cube. Therefore, the structure of the Segre embedding for \( n=5 \) follows the pattern of a 4-dimensional cube, as illustrated in the diagram Figure \ref{F:4-cube}.  For simplicity we have omitted the labels on the morphisms.

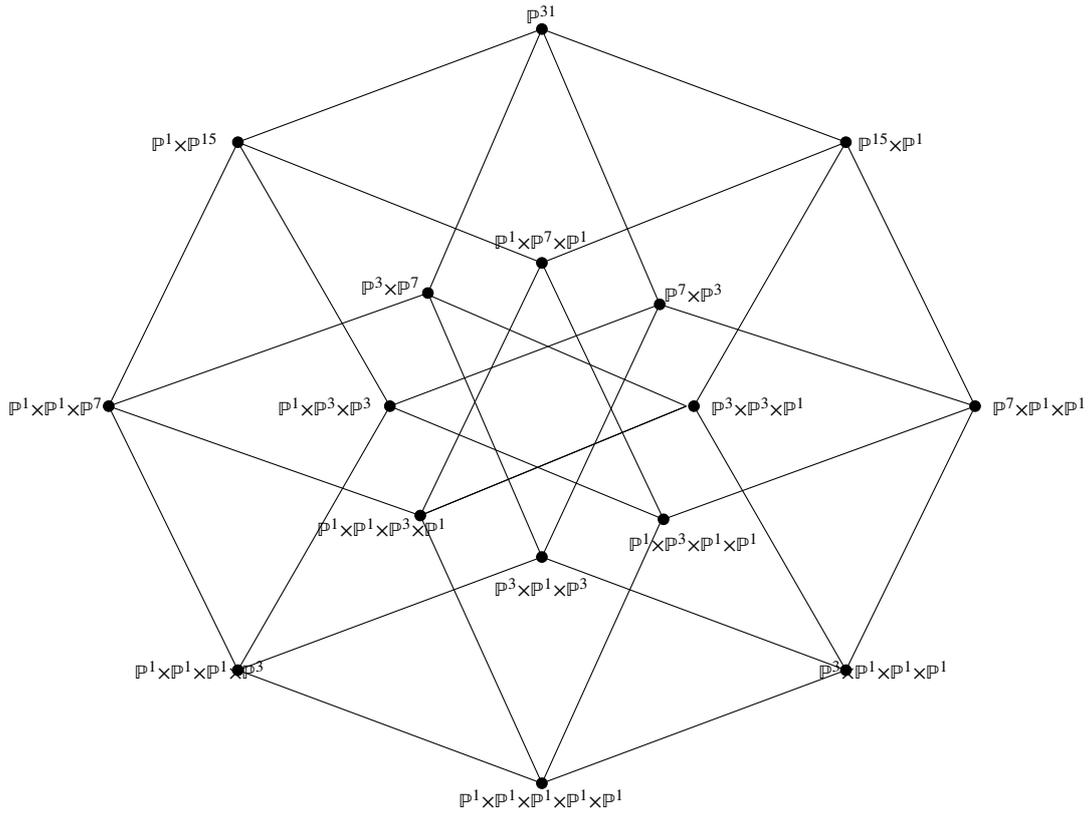
\begin{figure}
\begin{center} 
     \begin{tikzpicture}[scale=1.0]
  \draw[fill] (-0, 5) circle (.07cm);  \node (-0,5) at (-0,5.2) {$\scriptstyle \PP^{31}$};
  \node (-0,5) at (-0,5.6) {$$};
  \draw[fill] (-0, 1.9) circle (.07cm); \node (-0,19) at (-0,2.2) {$\scriptstyle \PP^1\times\PP^7\times\PP^1$};
  \draw[fill] (-0, -2) circle (.07cm);  \node (-0,-2) at (-0,-2.4) {$\scriptstyle \PP^3\times\PP^1\times\PP^3$};
  \draw[fill] (-0, -5) circle (.07cm);  \node (-0,-5) at (-0,-5.2) {$\scriptstyle \PP^1\times\PP^1\times\PP^1\times\PP^1\times\bP^1$};
   \node (-0,-5) at (-0,-5.6) {$$};
 \draw[fill] (-4, 3.5) circle (.07cm);  \node (-4,3.5) at (-4.7,3.5) {$\scriptstyle \PP^1\times\PP^{15}$};
 \draw[fill] (4, 3.5) circle (.07cm);  \node (4,3.5) at (4.6,3.5) {$\scriptstyle \PP^{15}\times\PP^1$};
 \draw[fill] (-4, -3.5) circle (.07cm);  \node (-4,5) at (-4.5,-3.5) {$\scriptstyle \PP^1\times\PP^1\times\PP^1\times \PP^3$};
 \draw[fill] (4, -3.5) circle (.07cm);  \node (4,-3.5) at (4.5,-3.5) {$\scriptstyle \PP^3\times\PP^1\times\PP^1\times \PP^1$};
  \draw[fill] (-5.7, 0) circle (.07cm);  \node (-5.7,0) at (-6.4,0) {$\scriptstyle \PP^1\times\PP^1\times \PP^7$};
 \draw[fill] (5.7, 0) circle (.07cm);  \node (5.7,0) at (6.2,0) {$\scriptstyle \quad \quad \PP^{7}\times \PP^1\times  \PP^1$};
  \draw[fill] (-2, 0) circle (.07cm);  \node (-2,0) at (-2.5,0) {$\scriptstyle \quad\quad \PP^1\times\PP^3\times \PP^3 \quad \quad\quad \quad$};
 \draw[fill] (2, 0) circle (.07cm);  \node (1.9,0) at (2.5,0) {$\scriptstyle  \quad\quad\PP^3\times\PP^3\times \PP^1$};
  \draw[fill] (-1.5, 1.5) circle (.07cm);  \node (-1.5,1.5) at (-1.99,1.6) {$\scriptstyle \PP^3\times\PP^7$};
\draw[fill] (1.6, -1.5) circle (.07cm);  \node (1.6,-1.5) at (2,-1.8) {$\scriptstyle \PP^1\times \PP^3\times\PP^1\times\PP^1$};
\draw[fill] (-1.6, -1.45) circle (.07cm);  \node (-1.6,-1.345) at (-2.1,-1.6) {$\scriptstyle \PP^1\times\PP^1\times \PP^3 \times \PP^1$};
\draw[fill] (1.55, 1.35) circle (.07cm);  \node (1.55,1.35) at (2,1.5) {$\scriptstyle \PP^7\times\PP^3$};

               \draw[black] (-0,5) -- (-4,3.5);
               \draw[black] (-0,5) -- (4,3.5);   
                   \draw[black] (-4,3.5) -- (-0,1.9);
                    \draw[black] (-0,1.9) -- (4,3.5);
                    \draw[black] (-0,-5) -- (-4,-3.5);
               \draw[black] (-0,-5) -- (4,-3.5);    
                   \draw[black] (-4,-3.5) -- (-0,-2);
                    \draw[black] (-0,-2) -- (4,-3.5); 
        \draw[black] (-5.7,0) -- (-4,-3.5);
          \draw[black] (5.7,0) -- (4,3.5);
         \draw[black] (5.7,0) -- (4,-3.5);
         \draw[black] (0,5) -- (-1.5,1.5);
        \draw[black] (2,0) -- (4,3.5);
         \draw[black] (2,0) -- (4,-3.5);
          \draw[black] (-5.7,0) -- (-4,3.5);
         \draw[black] (-2,0) -- (-4,-3.5);
         \draw[black] (-2,0) -- (-4,3.5);
         \draw[black] (-1.5,1.5) -- (-5.7,0);
          \draw[black] (1.6,-1.5) -- (5.7,0);
           \draw[black] (-1.5,1.5) -- (-0,-2);
           \draw[black] (1.6,-1.5) -- (-0,1.9);
                \draw[black] (0,-5) -- (-1.6,-1.45);
               \draw[black] (-5.7,0) -- (-1.6,-1.45);
                \draw[black] (1.9,0) -- (-1.6,-1.45);
                 \draw[black] (0,5) -- (1.55, 1.35);
                 \draw[black] (5.7,0) -- (1.55, 1.35);
                  \draw[black] (0,-5) -- (1.6, -1.45);
                   \draw[black] (0,-2) -- (1.55, 1.35);
                    \draw[black] (-2,0) -- (1.6, -1.5);
                    \draw[black] (-2,0) -- (1.55, 1.35);
                  \draw[black] (0,1.9) -- (-1.6, -1.45);
                  \draw[black] (1.9,0) -- (-1.6, -1.45);
\draw[black] (1.9,0) -- (-1.55, 1.5);

                  \end{tikzpicture}
                  \smallskip 
                  \caption{4-cube of Segre embeddings}\label{F:4-cube}
\end{center}
\end{figure}
\subsubsection*{General construction}
Our construction continues by a combinatorial argument. 
Let us consider the Segre embeddings, that is, embeddings of products of $n$ projective spaces into higher-dimensional projective spaces, given by the classical Segre maps. We seek to construct a bijection between the set of such embeddings and a combinatorial set $\mathcal{W}_{n-1}=\{w=(w_1,\cdots,w_{n-1})\, |\, w_i\in \{0,1\}\}$ of words of length $n-1$, formed from the alphabet $\{0,1\},$ encoding the sequence of Segre operations.

\, 
Let $(X_1,\cdots, X_n)$ be a sequence of projective spaces, each $X_i$ is of the form $\PP^{m_i}$, indexed by $i$. The iterated Segre embedding of these spaces into a higher-dimensional projective space follows a well-defined hierarchical construction:
\[
\PP^{m_1}\times \cdots \times \PP^{m_n}\hookrightarrow \PP^{N},\] where $N=\prod_{i=1}^n(m_i+1)-1$ and each embedding at an intermediate stage corresponds to the standard Segre map.

\subsubsection*{Encoding via Binary Words}
To each Segre embedding process, we associate a word 
$w$ of length $n-1$, whose letters belong to the alphabet 
$\{0,1\}$, as follows:
\begin{enumerate}
\item If a Segre embedding $ \PP^{m_1}\times \cdots \times \underbrace{\PP^{m_i}\times \PP^{m_{i+1}}}_i \times \cdots \times \PP^{m_n}$ is performed at step 
$i$, we set $w_i=1$. That is, if at the $i$-th stage, we embed two factors together via the Segre construction, this operation is encoded by a 
1 in the word $w$. In particular,  if initially the word was $$w=(w_1,\cdots, w_{i-1},\underbrace{0}_i,w_{i+1},\cdots w_{n-1}),$$ where $w_i\in\{0,1\}$ after the Segre embedding the new word becomes 
$$w=(w_1,\cdots, w_{i-1},\underbrace{1}_i,w_{i+1},\cdots w_{n-1}).$$
\item If no Segre embedding has occured for the $i$-th pair of projective spaces, we set 
the letter $w_i=0$. This means that the factors in the corresponding product $\PP^{m_1}\times \cdots \times \PP^{m_n}$ have not proceeded to a Segre embedding.
\end{enumerate}
Going back to the 4-cube, we give an example where the vertices are labelled by the binary words: 

\begin{center} 
     \begin{tikzpicture}[scale=1.0]
  \draw[fill] (-0, 5) circle (.07cm);  \node (-0,5) at (-0,5.2) {$\scriptstyle (1111)$};
  \node (-0,5) at (-0,5.6) {$\scriptstyle$};
  \draw[fill] (-0, 1.9) circle (.07cm); \node (-0,19) at (-0,2.2) {$\scriptstyle (0110)$};
  \draw[fill] (-0, -2) circle (.07cm);  \node (-0,-2) at (-0,-2.4) {$\scriptstyle (1001)$};
  \draw[fill] (-0, -5) circle (.07cm);  \node (-0,-5) at (-0,-5.2) {$\scriptstyle (0000)$};
   \node (-0,-5) at (-0,-5.6) {$\scriptstyle $};
 \draw[fill] (-4, 3.5) circle (.07cm);  \node (-4,3.5) at (-4.7,3.5) {$\scriptstyle (0111)$};
 \draw[fill] (4, 3.5) circle (.07cm);  \node (4,3.5) at (4.6,3.5) {$\scriptstyle (1110)$};
 \draw[fill] (-4, -3.5) circle (.07cm);  \node (-4,5) at (-4.5,-3.5) {$\scriptstyle (0001)$};
 \draw[fill] (4, -3.5) circle (.07cm);  \node (4,-3.5) at (4.5,-3.5) {$\scriptstyle (1000)$};
  \draw[fill] (-5.7, 0) circle (.07cm);  \node (-5.7,0) at (-6.4,0) {$\scriptstyle (0011)$};
 \draw[fill] (5.7, 0) circle (.07cm);  \node (5.7,0) at (6.2,0) {$\scriptstyle (1100)$};
  \draw[fill] (-2, 0) circle (.07cm);  \node (-2,0) at (-2.5,0) {$\scriptstyle (0101)$};
 \draw[fill] (2, 0) circle (.07cm);  \node (1.9,0) at (2.5,0) {$\scriptstyle (1010)$};
  \draw[fill] (-1.5, 1.5) circle (.07cm);  \node (-1.5,1.5) at (-1.99,1.6) {$\scriptstyle (1011)$};
\draw[fill] (1.6, -1.5) circle (.07cm);  \node (1.6,-1.5) at (2,-1.8) {$\scriptstyle (0100)$};
\draw[fill] (-1.6, -1.45) circle (.07cm);  \node (-1.6,-1.345) at (-2.1,-1.6) {$\scriptstyle (0010)$};
\draw[fill] (1.55, 1.35) circle (.07cm);  \node (1.55,1.35) at (2,1.5) {$\scriptstyle (1101)$};

               \draw[black] (-0,5) -- (-4,3.5);
               \draw[black] (-0,5) -- (4,3.5);   
                   \draw[black] (-4,3.5) -- (-0,1.9);
                    \draw[black] (-0,1.9) -- (4,3.5);
                    \draw[black] (-0,-5) -- (-4,-3.5);
               \draw[black] (-0,-5) -- (4,-3.5);    
                   \draw[black] (-4,-3.5) -- (-0,-2);
                    \draw[black] (-0,-2) -- (4,-3.5); 
        \draw[black] (-5.7,0) -- (-4,-3.5);
          \draw[black] (5.7,0) -- (4,3.5);
         \draw[black] (5.7,0) -- (4,-3.5);
         \draw[black] (0,5) -- (-1.5,1.5);
        \draw[black] (2,0) -- (4,3.5);
         \draw[black] (2,0) -- (4,-3.5);
          \draw[black] (-5.7,0) -- (-4,3.5);
         \draw[black] (-2,0) -- (-4,-3.5);
         \draw[black] (-2,0) -- (-4,3.5);
         \draw[black] (-1.5,1.5) -- (-5.7,0);
          \draw[black] (1.6,-1.5) -- (5.7,0);
           \draw[black] (-1.5,1.5) -- (-0,-2);
           \draw[black] (1.6,-1.5) -- (-0,1.9);
                \draw[black] (0,-5) -- (-1.6,-1.45);
               \draw[black] (-5.7,0) -- (-1.6,-1.45);
                \draw[black] (1.9,0) -- (-1.6,-1.45);
                 \draw[black] (0,5) -- (1.55, 1.35);
                 \draw[black] (5.7,0) -- (1.55, 1.35);
                  \draw[black] (0,-5) -- (1.6, -1.45);
                   \draw[black] (0,-2) -- (1.55, 1.35);
                    \draw[black] (-2,0) -- (1.6, -1.5);
                    \draw[black] (-2,0) -- (1.55, 1.35);
                  \draw[black] (0,1.9) -- (-1.6, -1.45);
                  \draw[black] (1.9,0) -- (-1.6, -1.45);
\draw[black] (1.9,0) -- (-1.55, 1.5);

                  \end{tikzpicture}
                  \smallskip 
\end{center}

Thus, we have a bijection between each step of a Segre embedding and an $n-1$ word formed from 0's and 1's. 
\subsubsection*{The Bijection}
This construction defines a bijection between:
\begin{itemize}
\item The set of possible Segre embeddings of a given sequence of $n$ projective spaces.
\item The set of possible binary words of length $n-1$, which describe precisely which Segre operations have been applied.
\end{itemize}

Let us now refine our perspective.

\subsubsection*{A Cubical Structure on the Set of Segre Embeddings}
The set of binary words of length $n-1$, which we have identified with the set of possible Segre embeddings, naturally forms a hypercube of dimension 
$n-1$. That is: \[\mathcal{W}_{n-1}=\{w=(w_1,\cdots,w_{n-1})\, |\, w_i\in \{0,1\}\}.\]

This set of words can be regarded as the vertex set of the 
$(n-1)$-dimensional cube $\mathcal{Q}_{n-1}$, where each coordinate of the word corresponds to a specific Segre embedding decision.

\subsubsection*{Edges and the Hamming Metric}
Two vertices (binary words) in $\mathcal{Q}_{n-1}$ are connected by an edge if and only if their corresponding words differ by exactly one letter. This corresponds precisely to the classical Hamming distance:
  \[d_{H}(w,w')=\sum_{i=1}^{n-1}|w_i-w'_i|.\]
  
Thus, two Segre embedding sequences are adjacent in their corresponding cube if and only if they differ at exactly one step of the embedding process, meaning that precisely only one Segre embedding has been performed on a product of projective spaces. This provides a natural graph-theoretic adjacency structure on the space of all possible Segre embeddings and therefore we have shown that the diagram of embeddings is a hypercube.  
  
  We finish the proof by induction. It is easy to conclude that for a Segre embedding 
  \[\PP^{m_1}\times \cdots \times \underbrace{\PP^{m_i}\times \PP^{m_{i+1}}}_i \times \cdots \times \PP^{m_{n+1}}\]  there exists an $n$-hypercube, which is in bijection and which  describes every set in the Segre embedding.

\end{proof}
\subsection{Pseudo code for generalized Segre embeddings}
\begin{algorithm}
\caption{Generate the Generalized Segre Embedding diagram}
\begin{algorithmic}[1]
\Procedure{Generate\, n\, Segre Embeddings}{n}
    \State \textbf{Input:} Integer \( n \) (dimension of the hypercube)
    \State \textbf{Output:} Set of vertices and edges of the \( n \)-hypercube
    \State

    \Comment{Step 1: Generate all vertices (binary words of length \( n \))}
    \State \( V \gets \{ \text{all binary strings of length } n \} \)

    \Comment{Step 2: Generate edges based on Hamming distance}
    \State \( E \gets \emptyset \)
    \ForAll{pairs \( (v_i, v_j) \) in \( V \)}
        \If{HammingDistance(\( v_i, v_j \)) \( = 1 \)}
            \State \( E \gets E \cup \{ (v_i, v_j) \} \)
        \EndIf
    \EndFor

    \State \textbf{return} \( (V, E) \)
\EndProcedure
\end{algorithmic}
\end{algorithm}

\begin{algorithm}
\caption{Hamming Distance Function}
\begin{algorithmic}[1]
\Function{HammingDistance}{ $s_1$, $s_2$ }
    \State \textbf{Input:} Two binary strings \( s_1, s_2 \) of equal length
    \State \textbf{Output:} Hamming distance \( d \) between \( s_1 \) and \( s_2 \)
    \State \( d \gets 0 \)
    \For{\textbf{each} corresponding pair \( (c_1, c_2) \) in \( s_1, s_2 \)}
        \If{ \( c_1 \neq c_2 \) }
            \State \( d \gets d + 1 \)
        \EndIf
    \EndFor
    \State \textbf{return} \( d \)
\EndFunction
\end{algorithmic}
\end{algorithm}

\vfill\eject
\section{Errors occurring and an attempt to possibly correct them} 

The quantum bits are particular sensitive to external phenomena, this can lead to errors. A difficult problem related to this is to understand how to detect errors and correct them. It can happen for instance that errors occur due to certain external phenomena, such as \emph{cosmic bit flips}. This produces errors also in the classical bit system. Cosmic bit flips refer to errors where individual bits are flipped due to external influences. Such errors form gateways to deeper mathematical considerations. We propose to use the Coxeter chambers and galleries method for error correcting, in the spirit of \cite{C18,C19}. For an introduction to the theory of chambers, refer to Bourbaki  \cite[Chap. IX, Sect. 5.2]{B}.

\, 

\subsection{Coxeter Groups and Reflection Arrangements}

The Coxeter group of type \( A_n \) is given by the symmetric group \( S_{n+1} \), which acts on \( \mathbb{R}^{n+1} \) by permuting coordinates. The corresponding \emph{root system} consists of the hyperplanes:
\[
H_{ij} = \{ x \in \mathbb{R}^{n+1} \mid x_i = x_j \}, \quad 1 \leq i < j \leq n+1.
\]
These hyperplanes decompose space into \emph{Coxeter chambers}, which are simplicial cones.

\subsection{Coxeter Chambers for \( A_n \)}
The chamber where $x_1 < x_2 < \dots < x_{n+1}$ serves as a fundamental chamber. All other chambers are obtained by applying group elements to it.
If we consider the hyperplane \( x_1 + x_2 + \dots + x_{n+1} = 0 \) in \( \mathbb{R}^{n+1} \), on which the Weyl group \( S_{n+1} \) acts naturally by permuting coordinates. This is an \( n \)-dimensional simplex.

\subsection{Transitive Action on Coxeter Chambers}

The Weyl group \( S_{n+1} \) acts simply transitively on the set of Coxeter chambers. That is, for any two chambers, there exists a unique element of \( S_{n+1} \) mapping one chamber to the other.

More precisely, if we consider the hyperplane \( x_1 + x_2 + \dots + x_{n+1} = 0\), for a chamber \( C_\omega \) associated with a permutation \( \omega \in S_{n+1} \), we have:
\[
C_\omega = \omega(C) = \{ x \in \mathbb{R}^{n+1} \mid x_{\omega(1)} < x_{\omega(2)} < \dots < x_{\omega(n+1)}, \quad x_1 + \dots + x_{n+1} = 0 \}.
\]
This establishes a bijective correspondence between \( S_{n+1} \) and the set of chambers.

\subsubsection{Pure separable product states}
We will consider the case of pure separable product states. 
\begin{thm}
Assume $|\psi\rangle$ is a pure separable product state. Suppose that the coefficients of $|\psi\rangle$ are perturbed by an error occurring, defining thus a new vector $|\psi'\rangle$. Then, if $|\psi'\rangle$ is a pure separable product state then either it belongs to the same Coxeter chamber $C$ as $|\psi\rangle$ on the Segre variety or it belongs to a different Coxeter chamber $C\neq C'$ such that $C'=gC$, where $g$ is an element of the group $S_{n+1}$. 
\end{thm}
\begin{proof}
\subsubsection*{Coxeter Group of Type $A$ and Its Chambers}
The Coxeter group of type 
$A_n$ corresponds to the symmetric group \(S_{n+1}\), which acts by permuting coordinates.
The chambers of this Coxeter group are formed from simplices, which are regions of space delimited by hyperplanes (also called walls) $H_{ij}$.
	These walls act as mirrors in the sense that reflections across them swap adjacent elements in a sequence. 
\subsubsection*{Step 1: Interpretation for the Segre Variety:}
Since the Segre variety is invariant under the action of this Coxeter group, it can be partitioned into chambers. The closure of a chamber is a fundamental domain, meaning that the entire space can be generated by acting on one chamber with the Coxeter group.
This claim can be already checked on the simplest examples:

Let us take the simplest Segre variety: $S=\{z_{00}z_{11}-z_{01}z_{10}\}$. By performing a permutation of the set of indices \{0,1\}, we have: 

\[\begin{aligned}
 &\omega\in S_2,&   0 \mapsto 1 \\
       &&1 \mapsto 0 \\
\end{aligned}\]
Therefore, the Segre variety $S_{\omega}$ becomes $z_{11}z_{00}-z_{10}z_{01}$. As we can see, nothing has changed. Therefore the Segre surface is invariant under the Coxeter group of type $A$.

We can check for this for the case discussed earlier arising from the Segre embedding $\bP^2\times \bP^2\hookrightarrow \bP^8$.
This is given by \[([x_0:x_1:x_2],[y_0:y_1:y_2])\mapsto [x_0y_0:x_0y_1:x_0y_2:x_1y_0:x_1y_1:x_1y_2:x_2y_0:x_2y_1:x_2y_2.]\]
The Segre variety is given by \[V(\{z_{ij}z_{kl}-z_{il}z_{kj}\, |\, 0\leq i,k\leq 2, \quad 0\leq k,l\leq 2 \}) \quad\text{in}\quad \bP^8\, \text{i.e.}\]
\[
\begin{aligned}
&a.& z_{00}z_{11} - z_{01}z_{10}=0, \\
&b.& z_{00}z_{12} - z_{02}z_{10}=0, \\
&c.&  z_{01}z_{12} - z_{02}z_{11}=0, \\
&d. & z_{00}z_{21} - z_{01}z_{20}=0, \\
&e.& z_{00}z_{22} - z_{02}z_{20}=0, \\
&f.& z_{01}z_{22} - z_{02}z_{21}=0, \\
&g.&  z_{10}z_{21} - z_{11}z_{20}=0, \\
&h.&  z_{10}z_{22} - z_{12}z_{20}=0, \\
&i.&  z_{11}z_{22} - z_{12}z_{21}=0.
\end{aligned}
\]
Let us chose a specific permutation \(\omega\in S_3\):
\[\begin{aligned}
 &\omega\in S_3,&   0 \mapsto 2 \\
       &&1 \mapsto 0 \\
         &&2 \mapsto 1 \\
\end{aligned}\]
This implies that the equation \((a)\) which is $z_{00}z_{11} - z_{01}z_{10}=0$ is mapped to $z_{22}z_{00} - z_{20}z_{02}=0$ which is equation \((e)\). More generally any permutation of  \(S_3\) will map a given equation from the set \(\{a,\cdots, i\}\) to another equation from the same set.  A proof of the invariance of the Segre variety under the Coxeter group of type $A$ is easy to outline. We sketch the idea below. 
Given the Segre embedding $\bP^m\times \bP^n\hookrightarrow \bP^{(m+1)(n+1)-1}$, it is given in coordinates by

\[([x_0:\cdots :x_m],[y_0:\cdots:y_n])\mapsto [x_0y_0:x_0y_1:x_0y_2:x_1y_0:\cdots :x_my_n.]\]
The Segre variety is given by \[V(\{z_{ij}z_{kl}-z_{il}z_{kj}\, |\, 0\leq i,k\leq m, \quad 0\leq l,j\leq n\}) \quad\text{in}\quad \bP^{(m+1)(n+1)-1}.\]

Consider, the pair $\omega=(p,q)$ residing within the Cartesian product $S_{m+1} \times S_{n+1}$, where \(S_{m+1}\) and \(S_{n+1}\) stand for the symmetric groups defined respectively for $m+1$ elements of the finite set \(\{0, \dots, m\}\) and $n+1$ elements of the set \(\{0, \dots, n\}\). The elements $p$ and $q$ are nothing but elements of the automorphism group of those finite sets i.e $p\in {\rm Aut}(\{0,\cdots,m\})$ and $q\in {\rm Aut}(\{0,\cdots,n\})$.

Now, let us consider the defining equations of the Segre variety: quadratic relations \(\{z_{ij}z_{kl} - z_{il}z_{kj} = 0\}\) defined for all $0\leq i,k\leq m$ and $0\leq l,j\leq n$. 
For any such equation $\{z_{ij}z_{kl}-z_{il}z_{kj}=0\}$,  the action of $\omega$ maps it to  the new equation \[\{z_{p(i)q(j)}z_{p(k)j(l)}-z_{p(i)q(l)}z_{p(k)q(j)}=0\},\]
where $0\leq p(i),p(k)\leq m, \quad 0\leq q(l),q(j)\leq n$.

Yet, by the very nature of \(p\) and \(q\) as elements of the automorphism group $p\in {\rm Aut}(\{0,\cdots,m\})$ and $q\in {\rm Aut}(\{0,\cdots,n\})$, the equation \[\{z_{p(i)q(j)}z_{p(k)j(l)}-z_{p(i)q(l)}z_{p(k)q(j)}=0\}\] can only exist within the generators of the ideal, given by the  relations 

$V(\{z_{ij}z_{kl}-z_{il}z_{kj}\, |\, 0\leq i,k\leq m, \, 0\leq l,j\leq n\})$. Therefore, the Segre variety is invariant under a Coxeter group of type $A$.

\subsubsection*{Step 2: Representation of the coordinates of $|\psi\rangle$ as a Point in a Coxeter Chamber}
Assume $|\psi\rangle$ correspond to a point on the Segre variety, lying within a given Coxeter chamber, denoted by $C$.
 The closure of such a Coxeter chamber can be considered as a  fundamental domain. 
 
 \subsubsection*{Step 3: Transitive Action on the Coxeter Chambers}
 The Coxeter group acts transitively on the chambers. That is, for any two chambers $C$ and $C'$ there exists an element $g$ of the Coxeter group such that $C=gC'$. This means that one chamber can be mapped to another by a suitable permutation.
 If the coefficients of $|\psi\rangle$  are modified, the new vector remains either in the same chamber $C$ or in another chamber $C'\neq C$. If $|\psi'\rangle$
 belongs to a different chamber, there must exist a group element $g$ that maps $C$ to $C'$. 

 \subsubsection*{Step 4: Conclusion}
 Therefore, this method allows us to keep track of possible occurring errors. Due to the simplicity of this method we can easily find an element which is the inverse of $ g$ within the Coxeter group. This ensures that we can recover the original chamber $C$.   

\end{proof}

\section{Maximally Entangled States and Their Geometry}
We end this paper by mentioning some important information concerning entanglement. A state is maximally entangled if and only if the 
matrix $c$ is unitary. Since an overall factor of this matrix is
irrelevant for the state one can conclude that the space of maximally entangled 
states is isomorphic to $SU(N)/\mathbb{Z}_N$. This happens to be an interesting submanifold 
of $ \mathbb{P}^{N^2-1}$: 
\begin{itemize}
\item it is \textit{Lagrangian} (a submanifold with vanishing symplectic form and half the dimension of the symplectic embedding space) and
\item \textit{minimal} (any attempt to move it will increase its volume).
\end{itemize}

When $N = 2$ we are looking at $SU(2)/{\mathbb Z}_2 = 
{\mathbb R}{\bf P}^3$. To see what this space looks like in the octant 
picture we observe that 
\begin{equation}
 \Gamma_{ij} = \left[ \begin{array}{cc} {\alpha} & {\beta} \\ 
- {\beta}^* & {\alpha}^* \end{array} \right] \hspace{5mm} \Rightarrow 
\hspace{5mm} Z^{\alpha} = ({\alpha}, {\beta}, -{\beta}^*, {\alpha}^*) \ .
\label{CijZ}
 \end{equation}

This space can be visualized in the octant representation as a geodesic passing through the center of the tetrahedral configuration representing entanglement classes.

\end{document}